\documentclass[11pt]{amsart}
\usepackage[foot]{amsaddr}
\usepackage{ifxetex}
\ifxetex
  \usepackage[no-math]{fontspec}
\else
\fi
\usepackage{amsmath}
\usepackage{amsfonts}
\usepackage{amssymb}
\usepackage{amsthm}
\usepackage{fullpage}
\usepackage[lining,semibold,type1]{libertine} 
\usepackage[T1]{fontenc} 
\usepackage{textcomp} 
\usepackage[varqu,varl]{inconsolata}
\usepackage[libertine,vvarbb]{newtxmath}
\usepackage[scr=rsfso]{mathalfa}
\usepackage{bm}
\usepackage{listings}
\usepackage{hyperref, color}
\hypersetup{colorlinks=true,citecolor=blue, linkcolor=blue, urlcolor=blue}
\usepackage[linesnumbered,boxed,ruled,vlined]{algorithm2e}
\usepackage{bm}
\usepackage{bbm}
\usepackage[numbers]{natbib}
\usepackage{xcolor}
\usepackage{enumerate}
\usepackage{enumitem}
\usepackage{tabularx}
\usepackage{array}
\usepackage{cleveref}
\usepackage{mathrsfs}

\newcolumntype{L}[1]{>{\raggedright\arraybackslash}p{#1}}
\newcolumntype{C}[1]{>{\centering\arraybackslash}m{#1}}
\newcolumntype{R}[1]{>{\raggedleft\arraybackslash}p{#1}}

\usepackage{makecell}

\usepackage{footnote}
\makesavenoteenv{tabular}

\newcommand{\e}{\mathrm{e}}

\renewcommand{\epsilon}{\varepsilon}

\newcommand{\bra}[1]{\langle #1 \rvert}
\newcommand{\ket}[1]{\lvert #1 \rangle}

\newtheorem{theorem}{Theorem}[section]

\newtheorem*{claim*}{Claim}

\newtheorem{fact}[theorem]{Fact}
\newtheorem{lemma}[theorem]{Lemma}

\newtheorem{corollary}[theorem]{Corollary}
\theoremstyle{definition}

\newtheorem{definition}[theorem]{Definition}

\newtheorem*{remark*}{Remark}

\newcommand{\observable}{\mathcal{O}}
\def\Pr{\mathop{\mathbf{Pr}}\nolimits}

\renewcommand{\emptyset}{\varnothing}

\newcommand{\norm}[1]{\left\Vert#1\right\Vert}

 \newcommand{\tuple}[1]{\left(#1\right)} 
 
 \newcommand{\tp}{\tuple}

\newcommand{\abs}[1]{\left\vert#1\right\vert}

\newcommand{\Tr}[1]{\mathbf{Tr}\left[#1\right]}
\newcommand{\Op}{\mathcal{O}}
\newcommand{\supp}{\mathrm{supp}}
\newcommand{\Hspace}{\+V}

\def\*#1{\boldsymbol{#1}} 
\def\+#1{\mathcal{#1}} 
\def\-#1{\mathrm{#1}} 
\def\^#1{\mathscr{#1}} 

\usepackage{todonotes}

\usepackage{xifthen}

\renewcommand{\Pr}[2][]{ \ifthenelse{\isempty{#1}}
  {\mathbf{Pr}\left[#2\right]} {\mathbf{Pr}_{#1}\left[#2\right]} } 
\newcommand{\E}[2][]{ \ifthenelse{\isempty{#1}}
  {\mathbf{\mathbf{E}}\left[#2\right]}
  {\mathbf{\mathbf{E}}_{#1}\left[#2\right]} }
  \newcommand{\Var}[2][]{ \ifthenelse{\isempty{#1}}
  {\mathbf{\mathbf{Var}}\left[#2\right]}
  {\mathbf{\mathbf{Var}}_{#1}\left[#2\right]} }

\newcommand{\relaxT}[2][]{
  \ifthenelse{\isempty{#2}}
  {t_{\mathrm{rel}}^{\mathrm{#1}}}
  {t_{\mathrm{rel}}^{\mathrm{#1}}(#2)}

}

\newcommand {\br} [1] {\ensuremath{ \left( #1 \right) }}

\title{Polynomial-Time Approximation of Zero-Free Partition Functions}
\author{Penghui Yao}
\author{Yitong Yin}
\author{Xinyuan Zhang}

\address[Penghui Yao, Yitong Yin, Xinyuan Zhang]{State Key Laboratory for Novel Software Technology, Nanjing University, 163 Xianlin Avenue, Nanjing, Jiangsu Province, China. \textnormal{E-mails: \url{pyao@nju.edu.cn}, \url{yinyt@nju.edu.cn}, \url{zhangxy@smail.nju.edu.cn}}}

\begin{document}
\maketitle

\begin{abstract}
%
Zero-free based algorithm is a major technique for deterministic approximate counting.
%
In Barvinok's original framework~\cite{barvinok-comb}, 
by calculating truncated Taylor expansions,
a quasi-polynomial time algorithm was given for estimating zero-free partition functions.
%
Patel and Regts~\cite{patel-polynomial} later gave a refinement of Barvinok's framework, which gave a polynomial-time algorithm for a class of zero-free graph polynomials that can be expressed as counting induced subgraphs in bounded-degree graphs.

In this paper, we give a polynomial-time algorithm for estimating classical and quantum partition functions specified by local Hamiltonians with bounded maximum degree, assuming a zero-free property for the temperature. 
Consequently, when the inverse temperature is close enough to zero by a constant gap, we have polynomial-time approximation algorithm for all such partition functions.
Our result is based on a new abstract framework that extends and generalizes the approach of Patel and Regts.

%
\end{abstract}

\section{Introduction}

Let $\Omega=[q]^V$ be a finite space of configurations, where $V$ is a set of $n$ variables.
Let $H_1,\ldots,H_m$ be a collection of local constraints, where each $H_j: \Omega\rightarrow \mathbb{C}$ is independent of all but a small subset of variables, and let $H=\sum_{j=1}^mH_j$.
The \emph{partition function} of the given system is defined by
\begin{align}
    Z_{H}(\beta) = \sum_{\sigma \in \Omega} \exp(-\beta \cdot H(\sigma)),\label{eq:classical-partition-function}
\end{align}
where the parameter $\beta$ is usually called the \emph{inverse temperature}.

The computational complexity of partition functions is one of the central topics in theoretical computer science,
which has been found wide applications in computational counting, combinatorics, and statistical physics.
To date, numerous algorithms as well as hardnesses of approximation for the partition functions of various systems have been established, to list a few~\cite{jerrum1993polynomial, salas1997absence, GJP03, jerrum2004polynomial, Wei06, stefankovic2009adaptive, sly2010computational, sly2012computational, LLY13, sinclair2014approximation, galanis2015inapproximability, liu2019correlation, bezakova2020inapproximability, chen2020optimal, bezakova2021complexity}.
The most important question here is, what property captures the approximability of partition functions.

It is widely believed that for various classes of partition functions of interests, the hardness of approximation is captured by the locus of complex zeros.
The study of the locus of complex zeros has a rich history in statistical physics, for example, in the famous Lee-Yang theorem~\cite{lee1952statistical}.
In computer science, the absence of complex zeros may imply efficient approximation algorithms for partition functions~\cite{barvinok-comb, patel-polynomial, regts2018zero, helmuth2020algorithmic, Liu2019Ising, liu2019correlation, peters2019conjecture, coulson2020statistical, guo2020soda, guo2020zeros, harrow2020classical, Shao2020Contraction, buys2021lee, chen2021spectral}.
This line of research was initiated by Barvinok's pioneering works~\cite{barvinok2015computing,barvinok2016computing,barvinok2017approximating,barvinok-comb,barvinok2017computing},
which used truncated Taylor expansions to approximate non-vanishing polynomials and established quasi-polynomial time approximations of partition functions with no complex zeros within a region.
Later in a seminal work of Patel and Regts~\cite{patel-polynomial},
this quasi-polynomial running time was  improved to polynomial time for a class of graph polynomials which can be expressed as induced subgraph sums in graphs of constant maximum degree.
And this polynomial-time framework was further extended by Liu, Sinclair and Srivastava \cite{Liu2019Ising} to hypergraph 2-spin systems with no complex-zeros for the external field.

For the quantum version,
several (classical) algorithms have been proposed, including \cite{kuwahara2020clustering, harrow2020classical, mann2021efficient}, to estimate the quantum generalization of the partition function~\eqref{eq:classical-partition-function} where $H$ is the \emph{Hamiltonian}. 
Yet an important question remains to answer is the polynomial-time approximability of the quantum or classic partition function in the form of~\eqref{eq:classical-partition-function}  assuming its zero-freeness.



\subsection{Our results.}
We show the polynomial-time approximability of zero-free \emph{quantum} partition functions.

%
Let $V$ be a set of $n$ \emph{sites} (also called \emph{vertices} or \emph{particles}).
Let $q\ge 2$ be an integer.
Throughout the paper, we assume that each site $u\in V$ is associated with a $q$-dimensional Hilbert space $\Hspace_u$, and let $\Hspace = \bigotimes_{u \in V} \Hspace_u$.
A \emph{Hamiltonian} $H$ is a Hermitian matrix in $\Hspace$.
The support of a Hamiltonian $H$, denoted by $\supp(H)$, is the minimal set of sites on which $H$ acts non-trivially.
Given a Hamiltonian $H$, $\exp(H)$ is defined by 
$\exp(H)=\sum_{\ell=1}^\infty \frac{1}{\ell!}H^\ell$,
and the \emph{partition function} $Z_H:\mathbb{C}\to\mathbb{C}$ induced by $H$ is defined as follows:
\begin{align}
\forall\beta\in\mathbb{C},\quad
    Z_H(\beta)
    \triangleq \Tr{\exp(-\beta  H)}.\label{eq:quantum-partition-function}
\end{align}

We are interested in partition functions induced by local Hamiltonians with bounded maximum degrees.

\begin{definition}[local Hamiltonian]\label{def:k-d-Hamiltonian}
A Hamiltonian $H\in\Hspace$ is said to be \emph{$k$-local}  if $H$ can be expressed as
\[
H=\sum_{j=1}^m H_j,
\]
where each $H_j$ acts non-trivially on at most $k$ sites, i.e.~$\abs{\supp(H_j)}\le k$.
A Hamiltonian $H\in\Hspace$ called a \emph{$(k,d)$-Hamiltonian} if $H$ is $k$-local and for every $v\in V$,
$\deg(v)\triangleq\abs{\left\{j\mid v\in \supp(H_j)\right\}}\le d$.
\end{definition}



Notice that if all $H_j$'s are diagonal, then $H$ is diagonal as well.
The quantum partition function $Z_H(\beta)$ then degenerates to the classical partition function defined in~\eqref{eq:classical-partition-function}.
Indeed, in such diagonal case we have
\[
Z_{H}(\beta)
= \sum_{\sigma \in [q]^{V}} \exp(-\beta\cdot H(\sigma)),
\]
where $H(\sigma)=\sum_{j=1}^mH_j(\sigma)$ and $H_j(\sigma)$ represents the $\sigma$-th diagonal entry of $H_j$. 
Since each $H_j$ is diagonal and acts non-trivially on subset $\supp(H_j)$ of at most $k$ sites, the value of $H_j(\sigma)$ is determines by the variables in $\supp(H_j)$. 
Hence the $Z_{H}(\beta)$ above is precisely the classical partition function defined in~\eqref{eq:classical-partition-function}.
%

The quantum partition functions encode rich information about quantum systems, e.g.~the free energy and ground state energy.
Meanwhile, the non-diagonal property, especially the non-commutativity of multiplication for non-diagonal matrices, imposes great challenges for the computation of partition functions.



We prove the following zero-free based approximability of quantum partition functions.
\begin{theorem}[\Cref{thm:quantum}, informal]\label{thm:quantum-informal}
Let $\Omega \subseteq \mathbb{C}$ be a ``well-shaped'' complex region (formalized by \Cref{def:good-region}) and $k,d\ge 1$ be constants.
There is a deterministic algorithm which takes a $(k,d)$-Hamiltonian $H$ on $n$ sites and a $\beta$ from interior of $\Omega$ as input, and outputs an estimation of the quantum partition function $Z_H(\beta)$ in polynomial time in $n$, if $Z_H$ satisfies the zero-free property such that $\abs{\log Z_H} \le \mathrm{poly}(n)$ on $\Omega$.
\end{theorem}
%
The formal statement (\Cref{thm:quantum}) is more general: it further takes into account the measurement of the quantum system.
Such generalization  may encode broader classes of partition functions, e.g.~the ones with external fields, and  also enable sampling from Gibbs state.
Following a recent major advance for quantum zero-freeness of Harrow \emph{et al}~\cite{harrow2020classical}, 
we give concrete applications (in \Cref{sec:applications}),
namely, polynomial-time algorithms for approximating the quantum partition function (\Cref{thm:quantum-concrete}) and sampling from the Gibbs state (\Cref{thm:quantum-sampler}) in a high-temperature regime (where $\beta$ is close to zero by a constant gap), improving the quasi-polynomial-time algorithms in~\cite{harrow2020classical}.
A  polynomial-time approximation of the quantum partition functions in a slightly bigger high-temperature regime was obtained in \cite{mann2021efficient} using the cluster expansion technique~\cite{helmuth2020algorithmic}, by transforming the quantum partition function to a polymer model and showing the convergence of the cluster expansion assuming high temperature. 


We prove polynomial-time approximability of the quantum partition function directly from a black-box property of zero-freeness, without further  restricting the parameters of the model.
Moreover, our result is proved in a new abstract framework, namely, functions specified by abstract neighborhood structures called \emph{dependency graphs}.
We prove the following general result.

\begin{theorem}[\Cref{cr:algo}, informal]\label{cr:algo-informal}
Suppose that functions $\{f_G\}$ specified by dependency graphs $G$ satisfies certain boundedness property of its Taylor coefficients (formalized in \Cref{def:good-function}).
Let $\Omega \subseteq \mathbb{C}$ be a ``well-shaped'' complex region. 
There is a deterministic algorithm which takes a dependency graph $G$ of $O(1)$ max-degree and $x$ from the interior of $\Omega$ as input, and outputs an estimation of $f_G(x)$ in polynomial time in size $n$ of $G$, if $f_G(0)$ is easy to compute and $f_G$ satisfies the zero-free property such that $\abs{\log f_G} \le \mathrm{poly}(n)$ on $\Omega$.
\end{theorem}

The abstract framework is described in \Cref{def:good-function}. 
As verified in \Cref{sec:approximation-zero-free}, 
our framework subsumes previous polynomial-time frameworks for zero-free based algorithms (\cite{patel-polynomial} and~\cite{Liu2019Ising}) as special cases,
and more importantly, it extends the previous frameworks to become compatible with infinite-degree polynomials and non-commutative systems, which are crucial for  quantum partition functions.

\section{Preliminaries}

\subsection{Local Hamiltonians}
Given a Hamiltonian $H$ in $\Hspace$, we use $\supp(H)$ to denote the \emph{support} of $H$, the minimal set of sites on which $H$ acts non-trivially.
Formally, if $S$ is the support of $H$, then $S$ is the minimal subset of $V$ satisfying that  $H=H_S\otimes I_{V\setminus S}$, where $H_S$ is a Hamiltonian in the space $\bigotimes_{v\in S}\Hspace_v$ and $I_{V\setminus S}$ is the identity matrix in the space $\bigotimes_{v\in V \setminus S}\Hspace_v$.
Readers may refer to~\cite{10.1561/0400000066,kitaev2002classical} for a thorough treatment.

\subsection{Basic facts about complex functions}

A complex-valued function $f:\Omega \rightarrow \mathbb{C}$ defined on a complex domain $\Omega \subseteq \mathbb{C}$ is called  \emph{holomorphic} if for every point $z\in\Omega$, the complex derivative exists in a neighborhood of $z$.
A holomorphic function $f:\Omega \rightarrow \mathbb{C}$ is infinitely differentiable and equals locally to its Taylor series.
A \emph{biholomorphic} function is a bijective holomorphic function whose inverse is also holomorphic
Furthermore, a function $f:\mathbb{C} \rightarrow \mathbb{C}$ is called an \emph{entire} function if it is holomorphic on $\mathbb{C}$.
A region $\Omega \subseteq \mathbb{C}$ is simply connected if $\overline{\mathbb{C}} \setminus \Omega$ is connected, where $\overline{\mathbb{C}}=\mathbb{C}\cup\{\infty\}$ denotes the extended complex plane.

The logarithm of a complex-valued function $f$, denoted by $g=\log f$, is a function such that $f(z)=\mathrm{e}^{g(z)}$.
For holomorphic function $f:\Omega \rightarrow \mathbb{C}\setminus\{0\}$ on simply connected region $\Omega \subseteq \mathbb{C}$,
such $\log f$ always exists (see e.g.~\cite{stein2010complex}).
Specifically, for an arbitrarily fixed pair $z_0,c_0 \in \mathbb{C}$ satisfying that $f(z_0) = \mathrm{e}^{c_0}$, we have
\begin{align}
    \forall z\in\Omega,\quad \log f(z) = \int_{P} \frac{f'(w)}{f(w)} \,dw\ + c_0,\label{eq:log-f}
\end{align}
where $P$ is an arbitrary path in $\Omega$ connecting $z$ and $z_0$.
%
%
%
Throughout the paper, we mainly deal with such holomorphic $f$ on simply connected $\Omega$ that $0\in\Omega$ and $f(0)\in\mathbb{R}^+$.
For such case, the definition of $\log f$ is uniquely determined by $z_0=0$ and the real $c_0=\ln(f(0))$.


\subsection{Approximation of non-vanishing function}
We now recap the polynomial interpolation technique of Barvinok~\cite{barvinok-comb} to estimate values of non-vanishing holomorphic functions.

For $b \in \mathbb{R}^+$, we use $\mathbb{D}_b$ to denote the complex disc of radius $b$ centered at the origin.
Formally,
\[
\mathbb{D}_b = \left\{z \in \mathbb{C} | \abs{z} < b\right\}.
\]
In particular,  let $\mathbb{D} = \mathbb{D}_1$ denote the unit disc.

For $\beta \in \mathbb{C}$ and $\delta \in \mathbb{R}^+$, we use $S_{\beta,\gamma}$ to denote $\delta$-strip of interval $[0,\beta]=\{t\beta\mid t\in[0,1]\}$.
Formally,
\[
    S_{\beta,\gamma} = \left\{z \in \mathbb{C} \mid \mathrm{dist}(z,[0,\beta])< \delta\right\}.
\]
where $\mathrm{dist}(\cdot,\cdot)$ denotes Euclidean distance.
It is clear that both $\mathbb{D}_b$ and $S_{\beta,\gamma}$ are simply connected.


The following is the key property of zero-freeness for complex-valued functions.
\begin{definition}[zero-freeness]
Let $M>0$ be finite positive real.
A holomorphic function $f$ on a simply connected region $\Omega\subseteq\mathbb{C}$ is \emph{$M$-zero-free} on $\Omega$ if
$\abs{\log f(z)} \le M$ for all $z \in \Omega$.
\end{definition}

Notice that the zero-freeness of $f$ on $\Omega$ implies that $f$ is non-vanishing on the same region. A definition of $\log f$ is assumed in the context when this concept is used.

For any polynomial $p \in \mathbb{C}[z]$ that does not vanish on $\mathbb{D}$, the polynomial $p$ automatically exhibits the above zero-freeness property with a bounded gap on $\mathbb{D}_{b}$ for any $b \in (0,1)$.
\begin{lemma}\label{lm:poly-bound}
    Let $p \in \mathbb{C}[z]$ be a polynomial of degree $d$, and let $b \in (0,1)$.
    If $p(z) \neq 0$ for all $z \in \mathbb{D}$, then $p$ is $M$-zero-free on $\mathbb{D}_b$ for $M=d \ln \frac{1}{1-b} + \abs{\log p(0)}$.
\end{lemma}
\begin{proof}
    Let $\zeta_1,\zeta_2,\ldots,\zeta_d \in \mathbb{C}$ denote the roots of polynomial $p$.
    For any $z \in \mathbb{D}_b$,
    \begin{align*}
        \abs{\log p(z)} = \abs{\log p(0) + \sum_{j=1}^d \log \tp{1-\frac{z}{\zeta_j}}} \le\abs{\log p(0)} - \sum_{j=1}^d \ln \tp{1-\abs{\frac{z}{\zeta_j}}} \le \abs{\log p(0)} + d \ln \frac{1}{1-b}.
    \end{align*}
The two inequalities are due to that all $\abs{\zeta_j}>1$ since $p(z) \neq 0$ for all $z \in \mathbb{D}$.
\end{proof}

The following lemma of Barvinok says that any holomorphic function on $\mathbb{D}$ can be approximated by its truncated Taylor expansion if it is zero-free on $\mathbb{D}$.

\begin{lemma}[\cite{barvinok-comb}]\label{lm:Barvinok}
    Let $g:\mathbb{D}\to\mathbb{C}$ be holomorphic and $M>0$.
    If $|g(z)|\le M$ for all $z\in\mathbb{D}$, then for any $z \in \mathbb{D}$ and any $m \in \mathbb{N}^+$,
    \begin{align*}
        \abs{g(z) - \sum_{k=0}^{m} \frac{g^{(k)}(0)}{k!} z^k} \le \frac{M}{\delta} (1-\delta)^{m+1},
    \end{align*}
    where $\delta = \mathrm{dist}(z,\partial \mathbb{D})$ represents the Euclidean distance between $z$ and the boundary of unit disc.
\end{lemma}
In particular, when the above lemma is applied to $g=\log f$ for some holomorphic $f:\mathbb{D}\to\mathbb{C}\setminus\{0\}$,
one can obtain a multiplicative approximation of $f$ on $\mathbb{D}$ assuming zero-freeness of $f$ on $\mathbb{D}$.
%
To make such approximation effective, we should be able to compute the Taylor coefficients of $g=\log f$.
%
%

The following \emph{Newton's identity} relates the Taylor coefficients of $g=\log f$ to those of $f$.

\begin{lemma}[Newton's identity]\label{lm:Newton}
    Let $f(z)=\sum_{k=0}^{+\infty} f_k z^k$ be an entire function such that $f(z) \neq 0$ for all $z \in \mathbb{D}$.
    Then $g(z)=\log f(z)=\sum_{k=0}^{+\infty} g_k z^k$ is well-defined on $\mathbb{D}$, and $$n g_n = nf_n - \sum_{k=1}^{n-1}kg_k f_{n-k}.$$
\end{lemma}

\begin{proof}
    By the definition of $g=\log f$, we have $f'=g'f$. Therefore,
    \[
        \begin{aligned}
            n f_n = \frac{1}{(n-1)!} f^{(n)}(0)
            = \frac{1}{(n-1)!} \sum_{k=0}^{n-1} \binom{n-1}{k} f^{(k)}(0)g^{(n-k)}(0)
            = \sum_{k=0}^{n-1} (n-k)g_{n-k}f_k
            = \sum_{k=1}^n kg_k f_{n-k}.
        \end{aligned}
    \]
\end{proof}

When the zero-free region is not unit disc, some preprocessing is needed.
The following polynomial transformation from $\mathbb{D}$ to any $S_{\beta,\gamma}$ is known.
\begin{lemma}[\cite{barvinok-comb}]\label{lm:poly}
    For any $\beta \in \mathbb{C}$, $\delta \in (0,1)$, there is an explicitly constructed polynomial $p_{\beta,\delta}$ of degree $d=d(\beta,\delta)$ satisfying
    \begin{itemize}
        \item $p_{\beta,\delta}(0)=0$ and $p_{\beta,\delta}(1-\delta_0)=\beta$ for some $\delta_0 \in (0,1)$;
        \item $p_{\beta,\delta}(\mathbb{D}) \subseteq S_{\beta,\gamma}$;
    \end{itemize}
\end{lemma}
The proof of \Cref{lm:poly} is deferred to \Cref{sec:appendix-poly}.

\section{Approximation of Zero-Free Holomorphic Function}\label{sec:approximation-zero-free}
We now introduce an abstraction for partition functions,
namely, multiplicative holomorphic functions specified by a class of abstract structures called dependency graphs.

A \emph{dependency graph} is a vertex-labeled graph $G=(V,E,\+L)$,
where $(V,E)$ is an undirected simple graph,
and $\+L=(L_v)_{v\in V}$ assigns each vertex $v\in V$ a label $L_v$.
Two labeled graphs $G=(V,E,\+L)$ and $G'=(V',E',\+L')$ are isomorphic if there is a bijection $\phi:V\to V'$ such that the two simple graphs $(V,E)$ and $(V',E')$ are isomorphic under $\phi$ and $L_v=L'_{\phi(v)}$ for all $v\in V$.
Furthermore, we say that two labeled graphs $G=(V,E,\+L)$ and $G'=(V',E',\+L')$ are disjoint if $V\cap V'=\emptyset$.
%
A family $\+G$ of dependency graphs is called \emph{downward-closed} if for any $G=(V,E,\+L) \in \+G$ and any $S\subseteq V$ we have $G[S] \in \+G$, where $G[S]$ stands for the subgraph of $G$ induced by subset $S\subseteq V$ preserving labels.

We use $f_{\cdot}$ to denote an operator that maps each dependency graph $G$ in $\+G$ to an entire function $f_G:\mathbb{C}\to\mathbb{C}$ (i.e.~$f_G$ is holomorphic on $\mathbb{C}$),
such that $f_G$ gives the same entire function for isomorphic dependency graphs $G$.
Such an $f_{\cdot}$ is \emph{multiplicative} if for any $G$ that is disjoint union of $G_1,G_2$, we have $f_G=f_{G_1}f_{G_2}$.
\begin{definition}[boundedness]\label{def:good-function}
Let $\+G$ be a downward-closed family of dependency graphs.
Let $\alpha,\beta \ge 1$.
A multiplicative $f_{\cdot}$ is called \emph{$(\alpha,\beta)$-bounded} on $\+G$ if
for any $G=(V,E,\+L)\in\+G$, we have $f_G(0)\in\mathbb{R}^+$ and
        \begin{align*}
            f_G(z) = f_G(0) + \sum_{\ell=1}^{+\infty} \tp{\sum_{S \subseteq V} \lambda_{G[S],\ell} }z^\ell,
        \end{align*}
where the complex coefficients $(\lambda_{H,\ell})_{H\in\+G, \ell\in\mathbb{N}^+}$ are invariant for isomorphic dependency graphs $H$,
and satisfy that $\lambda_{H,\ell}\neq 0$ only if $|V_H|\le\alpha \ell$,
and each $\lambda_{H,\ell}$ can be calculated within  $\beta^\ell\cdot\mathrm{poly}(\ell)$ time.
\end{definition}

For $(\alpha,\beta)$-bounded $f_{\cdot}$, it always holds that $f_G(0)\in\mathbb{R}^+$.
Then we always fix the definition of $\log f$ to be the one uniquely defined by Eq.\eqref{eq:log-f} with $z_0=0$ and $c_0=\ln(f(0))$ being real.
Such  $\log f$ is well defined within a neighborhood of the origin.

As we will formally verify in~\Cref{sec:generalize-BIGCP}, this notion of bounded holomorphic functions specified by dependency graphs generalizes the bounded induced graph counting polynomials (BIGCPs) of Patel and Regts~\cite{patel-polynomial}.
A major difference here is that $f_G$ may not be a polynomial of finite degree.
%



We show that for $(\alpha,\beta)$-bounded $f_{\cdot}$, the approach of Patel and Regts~\cite{patel-polynomial} based on Newton's identity and local enumeration of connected subgraphs can efficiently compute Taylor coefficients of $\log f_G$, even though the function $f_G$ can now encode problems far beyond counting subgraphs.

\begin{theorem}[efficient coefficient computing]
\label{thm:main}
%
Let $\+G$ be a downward-closed family of dependency graphs, and $f_{\cdot}$ be $(\alpha,\beta)$-bounded on $\+G$ for $\alpha,\beta\ge 1$.
There exists a deterministic algorithm which given any $G\in\+G$ and $\ell \in \mathbb{N}^+$ as input,
outputs the  $\ell$-th coefficient of the Taylor series of $\log f_G$ at the origin in time $\widetilde{O}\tp{n(8\mathrm{e}\beta\Delta)^{\alpha\ell}}$, where $n$ is the number of vertices, $\Delta$ is the maximum degree of $G$, and  $\widetilde{O}(\cdot)$ hides  $\mathrm{poly}(\Delta, \ell, \log(n))$.
%
\end{theorem}

The theorem will be proved in \Cref{section-efficient-coefficient-computing}.


Due to Riemann mapping theorem in complex analysis,
for any proper and simply connected region $\Omega \subset \mathbb{C}$ and any point $z_0 \in \Omega$,
there is a biholomorphic function $h$ from $\mathbb{D}$ to $\Omega$ such that $h(0)=z_0$.
%
We are interested in those good regions $\Omega \subseteq \mathbb{C}$ such that a transformation $h$ from $\mathbb{D}$ to $\Omega$ does not only exist but also has efficiently computable Taylor coefficients.
\begin{definition}[good region]\label{def:good-region}
Let $\gamma\ge 1$.
A simply connected region $\Omega \subseteq \mathbb{C}$ is \emph{$\gamma$-good}
if $0\in\Omega$ and given any $x \in \Omega$, there exists a holomorphic function $h_x$ on $\mathbb{D}$ along with a $z_x\in\mathbb{D}$   such that:
    \begin{enumerate}
        \item $h_x(\mathbb{D})\subseteq \Omega$, $h_x(0)=0$ and $h_x(z_x) = x$;
        \item 
        for every $\ell\in\mathbb{N}^+$, the $\ell$-th Taylor coefficient $h_{x,\ell}$ of $h_x$ at 0 can be determined in $\gamma^\ell \mathrm{poly}(\ell)$ time.
    \end{enumerate}
Given a  $\gamma$-good region $\Omega \subseteq \mathbb{C}$ and $\delta\in(0,1)$, we use $\Omega_{\delta}$ to denote the set of all $x\in\Omega$ with $z_x\in\mathbb{D}_{1-\delta}$.
\end{definition}

Any convex region is $1$-good, given access to an oracle that determines the distance to its boundary, which will be formally proved in \Cref{sec:appendix-poly}.
\begin{fact}\label{fact:convex}
Let $\Omega \subseteq \mathbb{C}$ be a convex region.
Suppose that for any $z\in\Omega$, $\mathrm{dist}(z,\partial \Omega)$ can be determined in $O(1)$ time.
Then $\Omega$ is $1$-good.
\end{fact}


Applying our \Cref{thm:main} to $\log f_G$, combining with Barvinok's approximation (\Cref{lm:Barvinok}) and our notion of good region, we obtain the following theorem for multiplicative approximation of zero-free $f_G$'s.

\begin{theorem}
[efficient $\epsilon$-approximation]
\label{cr:algo}
Let $\alpha,\beta,\gamma\ge 1$.
Let $\+G$ be a downward-closed family of dependency graphs, and $f_{\cdot}$ be $(\alpha,\beta)$-bounded on $\+G$.
Let $\Omega \subset \mathbb{C}$ be a $\gamma$-good region.
%

For any $\delta\in(0,1)$, there is a deterministic algorithm which takes $G \in \+G$, $x\in \Omega_{\delta}$ and an error bound $\epsilon \in (0,1)$  as input,
and outputs an estimation $\widehat{f_G(x)}$ of $f_G(x)$ within $\epsilon$-multiplicative error:
\[
1-\epsilon \le \frac{\abs{\widehat{f_G(x)}}}{\abs{f_G(x)}}\le 1+\epsilon,
\]
in $\widetilde{O}\tp{n\tp{\frac{M}{\delta \epsilon}}^C}$ time with $C=\frac{\alpha}{\delta}\ln (8\mathrm{e}\Delta(\beta+\gamma))$, where $\Delta$ is the maximum degree of $G$,
if provided that $f_G$ is $M$-zero-free on $\Omega \subset \mathbb{C}$ and the value of $f_G(0)$ is also provided to the algorithm.
\end{theorem}

Note that in above theorem, the zero-freeness property can be verified on any particular class of dependency graphs, although the boundedness property should be guaranteed on a downward-closed family.
When applying this theorem, the value of $f_G(0)$ is usually trivial to compute (e.g.~$f_G(0)=1$), and the $M$-zero-freeness is usually established for some $M=\mathrm{poly}(n)$ (e.g.~$M=O(n)$) on $n$-vertex dependency graphs.
In such typical cases, the running time in \Cref{cr:algo} is bounded as $\tp{\frac{n}{\delta \epsilon}}^{O\tp{\frac{1}{\delta}\log \Delta}}$.

\begin{proof}[Proof of \Cref{cr:algo}]
Let $h=h_x$ be a holomorphic function that transforms $\mathbb{D}$ to the $\gamma$-good region $\Omega$ with $h(z_x)=x$, where $z_x\in\mathbb{D}_{1-\delta}$ since $x\in\Omega_{\delta}$.
%
And let $f^{h}_G = f_G \circ h$.

First, observe that $f_G^h$ is $M$-zero-free on $\mathbb{D}$, because $\abs{\log f_G^h(z)}=\abs{\log f_G(h(z))} \le M$ holds for all $z \in \mathbb{D}$ since $h(\mathbb{D}) \subseteq \Omega$ and $f_G$ is $M$-zero-free on $\Omega$.
Then by \Cref{lm:Barvinok}, for any $z \in \mathbb{D}$, the difference between $\log f^{h}_G(z)$ and the truncated Taylor expansion at $0$ is bounded by
    \begin{align}\label{eq:truncate}
        \abs{\log f^h_G(z) - \sum_{k=0}^{m} \frac{\tp{\log f_G^{h}}^{(k)}(0)}{k!} z^k} \le \frac{M}{\delta} (1-\delta)^{m+1} < \epsilon,
    \end{align}
for $m = \lceil \frac{1}{\delta} \ln \frac{M}{\delta\epsilon}\rceil$.

It remains to verify that $f^h_\cdot$ is $(\alpha,\beta+\gamma)$-bounded on $\+G$.
By \Cref{thm:main}, this will prove the theorem.

For $\ell\in\mathbb{N}^+$, let $h_k^{(\ell)}$ denote the $k$-th Taylor coefficient of $h(z)^\ell$ at $z=0$. Since $h(0)=0$, we have
\[
h(z)^\ell
= \tp{\sum_{k=1}^{+\infty} h_{k} z}^\ell
        = \sum_{k=\ell}^{+\infty} h_{k}^{(\ell)} z^k
\]
Since $f_G$ is $(\alpha,\beta)$-bounded and $h(0)=0$, we have
    \begin{align*}
        f^h_G(z) &= f_G(0) + \sum_{\ell=1}^{+\infty} \tp{\sum_{S \subseteq V} \lambda_{G[S],\ell}}h(z)^\ell\\
         &= f_G(0) + \sum_{\ell=1}^{+\infty} \tp{\sum_{S \subseteq V} \lambda_{G[S],\ell}} \tp{\sum_{k=\ell}^{+\infty} h_{k}^{(\ell)} z^k}\\
         &= f_G(0) + \sum_{k=1}^{+\infty} \sum_{S \subseteq V} \tp{\sum_{\ell=1}^k h_{k}^{(\ell)} \lambda_{G[S],\ell}} z^k\\
         &= f_G(0) + \sum_{k=1}^{+\infty} \tp{\sum_{S \subseteq V} \lambda^h_{G[S],k} }z^k,
    \end{align*}
where $\lambda^h_{H,k}$ for any $H\in\+G$ and $k\in\mathbb{N}^+$ is defined as
\[
\lambda^h_{H,k}\triangleq \sum_{\ell=1}^k h_{k}^{(\ell)} \lambda_{H,\ell}.
\]
Clearly, $\lambda^h_{H,k}=0$ if $|V_H|>\alpha k$, where $V_H$ denotes the vertex set of $H$, since $\lambda_{H,\ell}=0$ if $|V_H|>\alpha \ell$.
And it can be verified that any $\lambda^h_{H,k}$ can be determined  within $(\beta+\gamma)^k \mathrm{poly}(k)$ time.
This is because within $\beta^k\mathrm{poly}(k)$ time, one can list all $\lambda_{H,1},\ldots,\lambda_{H,k}$, and for each $1\le\ell\le k$, $h_{k}^{(\ell)}$ is just the coefficient of $z^k$ in $\tp{h_1z+h_2z^2+\cdots +h_k z^k}^\ell$, which can be calculated in $\mathrm{poly}(k)$ time given all $h_1,\ldots,h_k$, which can be listed beforehand in $\gamma^k\mathrm{poly}(k)$ time.
Overall, this takes at most $(\beta+\gamma)^k\mathrm{poly}(k)$ time.

    Therefore, $f^h_{\cdot}$ is $(\alpha,\beta+\gamma)$-bounded. 
\end{proof}

\subsection{Quantum partition functions}
We formally prove \Cref{thm:quantum-informal}.
Recall the definition of quantum partition function $Z_H$  in~\eqref{eq:quantum-partition-function}.
We extend this definition by considering measurement.

Recall that $\Hspace = \bigotimes_{u \in V} \Hspace_u$ where $V$ is the set of $n$ sites and each $\Hspace_u$ is a $q$-dimensional Hilbert space.
A \emph{measurement} $\Op$ is a positive operator in $\Hspace$.
The quantum partition function induced by Hamiltonian $H$ under measurement $\Op$, both in $\Hspace$, is defined by
\begin{align}
Z_{H,\Op}(\beta)
\triangleq \Tr{\exp(\beta H)\Op}. \label{eq:quantum-pf1}
\end{align}
Furthermore, a measurement $\Op$ is \emph{tensorized} if $\Op=\bigotimes_{v \in V} \Op_v$ where $\supp(\Op_v)=\{v\}$.

We show that under tensorized measurement $\Op$, the quantum partition functions $Z_{H,\Op}$ defined by local Hamiltonians with $O(1)$ maximum degree are $(1,O(1))$-bounded.
Together with \Cref{cr:algo} we obtain the following theorem.

\begin{theorem}\label{thm:quantum}
Let $\Omega \subset \mathbb{C}$ be a $\gamma$-good region for $\gamma\ge 1$.
%
%
For any $\delta\in(0,1)$, there is a deterministic algorithm
such that given any $(k,d)$-Hamiltonian $H$ and tensorized measurement $\+O$,
provided that $\frac{1}{\Tr{\+O}}Z_{H,\+O}$ is $M$-zero-free on $\Omega$,
for any temperature $\beta \in \Omega_{\delta}$ and error bound $\epsilon \in (0,1)$,
the algorithm outputs an estimation of  $Z_{H,\+O}(\beta)$ within $\epsilon$-multiplicative error
in $\widetilde{O}\tp{n\tp{\frac{M}{\delta \epsilon}}^C}$ time with $C = \frac{1}{\delta}\ln\tp{8\mathrm{e}d\tp{2q^{3k}+\gamma}}$.
%
\end{theorem}

Note that when the measurement $\+O$ is the identity, $Z_{H,\+O}$ is precisely the partition function $Z_H$, which implies \Cref{thm:quantum-informal}.
As discussed in the introduction,
\Cref{thm:quantum} covers all classical partition functions (with or without external field) when temperature is the complex variable.

Following a recent work of Harrow, Mehraban and Soleimanifar~\cite{harrow2020classical}, \Cref{thm:quantum} gives polynomial-time approximations of quantum partition functions defined by local Hamiltonians with $O(1)$ maximum degree when  the inverse temperature $\beta$ is close to zero.
And following a standard routine of self-reduction, in the same regime, we have a polynomial-time approximate sampler from the quantum Gibbs state after the measurement in the computational basis.
These applications are given in \Cref{sec:applications}.
%

\begin{proof}[Proof of \Cref{thm:quantum}]
Given a $(k,d)$-Hamiltonian $H=\sum_{j=1}^m H_j$, we can construct a dependency graph $G_{H}=(U,E,\+L)$ as follows:
    \begin{enumerate}
        \item $U=[m]$ is the vertex set;
        \item $E=\left\{\{x,y\}\in {\binom{U}{2}} \mid \mathrm{supp}(H_x) \cap \mathrm{supp}(H_y) \neq \emptyset\right\}$;
        \item for any $x \in U$, its label is given by $L_x = H_x$.
    \end{enumerate}
Let $\+G_{k}$ denote the family of all such $G_H$ where $H$ is a $k$-local Hamiltonian.
It is obvious that such $\+G_{k}$ is downward-closed.

Let $\+O=\bigotimes_{v \in V} \+O_v$ be a tensorized measurement.
For any $G=G_H\in\+G_{k}$, where $H=\sum_{j=1}^m H_j$,
define:
\begin{align*}
        f_G(z) = \frac{1}{\Tr{\+O}}Z_{H,\+O}(z)=\frac{1}{\Tr{\+O}}\Tr{\exp\tp{-z \sum_{j =1}^m H_j} \+O}.
\end{align*}

The rest of the proof verifies that such $f_{\cdot}$ is $(1,2q^{3k})$-bounded on $\+G_{k}$,
which is sufficient to prove the theorem by \Cref{cr:algo}.


We first verify that such $f_{\cdot}$ is multiplicative.
For any $G=(U,E,\+L) \in \+G_k$ that is the disjoint union of $G_1=(U_1,E_1,\+L_1)$ and $G_2=(U_2,E_2,\+L_2)$, there exists a bipartition $V=V_1\uplus V_2$ such that $\supp(H_x) \subseteq V_1$ for all $x \in U_1$ and $\supp(H_y) \subset V_2$ for all $y \in U_2$.
Let $H_{U_i} = \sum_{x \in U_i} H_x$ for $i=1,2$.
We have,
    \begin{align*}
        f_G(z) &= \frac{1}{\mathbf{Tr}[\+O]} \mathbf{Tr}\left[\exp\tp{-z (H_{U_1}+H_{U_2})}\+O\right]\\
        & = \frac{1}{\mathbf{Tr}[\+O]} \mathbf{Tr}[\exp(-z H_{U_1}) \exp(-z H_{U_2})\+O]\\
        &= \frac{1}{\mathbf{Tr}[\+O]} \mathbf{Tr}_{V_1}  \left[\exp(-z H_{U_1}) \bigotimes_{v \in V_1} \+O_v\right] \mathbf{Tr}_{V_2}\left[\exp(-z H_{U_2}) \bigotimes_{v \in V_2} \+O_v\right]\\
        & = \frac{1}{\mathbf{Tr}^2[\+O]} \mathbf{Tr}[\exp(-z H_{U_1}) \+O] \mathbf{Tr}[\exp(-z H_{U_2})\+O]\\
        & = f_{G_1}(z)  f_{G_2}(z)
    \end{align*}
 Here the subscripts $V_1,V_2$ in $\mathbf{Tr}_{V_1},\mathbf{Tr}_{V_2}$ indicates the sets of sites that the operators act on. Therefore, $f_{\cdot}$ is multiplicative.

For any $G=(U,E,\+L)\in\+G_{k}$ and any $\ell\in\mathbb{N}^+$, define
    \begin{equation}\label{eqn:lambda}
        \lambda_{G,\ell}=\frac{1}{\ell!} \frac{1}{\mathbf{Tr}[\+O]}\sum_{f:[\ell]\stackrel{\mathrm{onto}}{\to}U}\Tr{\tp{\prod_{j=1}^\ell H_{f(j)}}\Op}.
    \end{equation}
    Observe that $\lambda_{G,\ell} = 0$ if $\abs{U} > \ell$ as there is no surjection from $[\ell]$ to $U$.
    Moreover, for $H=\sum_{x\in U}H_x$,
    \begin{align*}
        f_G(z)&=1+ \sum_{\ell=1}^{+\infty} \frac{\beta^\ell}{\ell!}\Tr{H^\ell \Op}
        =1+\sum_{\ell=1}^{+\infty} \frac{z^\ell}{\ell!} \sum_{ x_1,x_2,\ldots,x_\ell \in U} \Tr{\tp{\prod_{j=1}^\ell H_{x_j}} \Op}
        =1+\sum_{\ell=1}^{+\infty} \br{\sum_{S \subseteq U} \lambda_{G[S],\ell}} z^{\ell}. 
    \end{align*}

    It remains to show that $\lambda_{G,\ell}$ can be determined within $(2q^{3k})^\ell \mathrm{poly}(\ell)$ time.
    %

Fix a $G=(U,E,\+L) \in \+G_{k}$.
For any $S \subseteq U$, define
    \begin{align*}
        H_{S,\ell}\triangleq \sum_{f:[\ell]\stackrel{\mathrm{onto}}{\to}S}\prod_{j=1}^\ell H_{f(j)}.
    \end{align*}
Clearly, 
$\lambda_{G,\ell}=\frac{1}{\ell!}\frac{1}{\mathbf{Tr}[\+O]}\Tr{H_{U,\ell}\Op}$.
    Moreover, the following recurrence holds for $H_{S,\ell}$:
    \begin{align}
        H_{S,\ell}=\sum_{j \in S}H_j\tp{H_{S,\ell-1}+H_{S\setminus\{j\},\ell-1}},\label{eq:DP-H-S-ell}
    \end{align}
    where the boundary cases are given by 
    $H_{\emptyset,0}=I$, and $H_{S,\ell}=\boldsymbol{0}$ (the 0-matrix) if $\ell<|S|$, or $S=\emptyset$ but $\ell>0$.
    Note that $H_{S,\ell}$ acts non-trivially on at most $k\abs{S}$ sites, where each site corresponds to a $q$-dimensional Hilbert space, thus $H_{S,\ell}$ can be represented as a matrix of size at most $q^{k|S|}\times q^{k|S|}$ and the recursion step~\eqref{eq:DP-H-S-ell} can be evaluated in time $O(|S|q^{3k|S|})$.
    Therefore, for any $S\subseteq U$ that $1\le |S|\le \ell$,
    $H_{S,\ell}$ can be computed in time $O(2^{|S|}\ell |S| q^{3k|S|})=O(\ell^22^{\ell} q^{3k\ell})$ by a dynamic programming that constructs a $2^{S}\times [\ell]$ table according to the recurrence~\eqref{eq:DP-H-S-ell}.
    And finally, $\lambda_{G,\ell}=\frac{1}{\ell!} \frac{1}{\mathbf{Tr}[\+O]}\Tr{H_{U,\ell}\Op}$ can be computed from $H_{U,\ell}$ in  $O(q^{3k\ell})$ time because $H_{U,\ell}$ acts non-trivially on at most $k|U|\le k\ell$ sites and $\Op$ is tensorized.
\end{proof}

\subsection{Induced subgraph counting}\label{sec:generalize-BIGCP}
Our framework (\Cref{def:good-function}) subsumes bounded induced graph counting polynomials (BIGCP) defined by Patel and Regts~\cite{patel-polynomial}.

A BIGCP $p_\cdot$ defines multiplicative graph polynomials $p_G$ for all graphs $G=(V,E)$.
Moreover, there exists integer $\alpha\ge 1$ and sequence $\lambda_{H,\ell}$ of complex values such that the following conditions are satisfied.
\begin{enumerate}
    \item For any graph $G=(V,E)$, $p_G$ can be expressed as
    \begin{align*}
        p_{G}(z) = 1+ \sum_{\ell=1}^{m(G)} \tp{\sum_{\substack{H=(V_H,E_H)\\\abs{V_H} \le \alpha \ell}} \lambda_{H,\ell}\cdot \mathrm{ind}(H,G)} z^\ell,
    \end{align*}
    where $\mathrm{ind}(H,G)$ represents the number of induced subgraphs $G[S]$, $S\subseteq V$,  isomorphic to $H$.
    \item $\lambda_{H,\ell}$ can be determined in $O(\beta^\ell)$ time for some constant $\beta \ge 1$.
\end{enumerate}

For any $G=(V,E)$, we define a dependency graph $G^*=(V,E,\+L)$ where $\+L$ labels every $v\in V$ with a trivial symbol $*$.
Let $\+G$ denote the family of all $G^*$, which is clearly downward-closed.
We define $f_{G^*}=p_G$.


Note that
\begin{align*}
    \sum_{\substack{H=(V_H,E_H)\\\abs{V_H} \le \alpha \ell}} \lambda_{H,\ell} \mathrm{ind}(H,G) = \sum_{\substack{S \subseteq V\\ \abs{S} \le \alpha \ell}} \lambda_{G[S],\ell} .
\end{align*}
Therefore, any BIGCP $p_{\cdot}$ corresponds to an $f_\cdot$ that is $(\alpha,\beta)$-bounded on $\+G^*$.

\subsection{Boolean CSP with external field}\label{sec:3-2}
A Boolean-variate constraint satisfaction problem (Boolean CSP) is specified by a $\Phi=(V,E,\boldsymbol{\phi})$, where $H=(V,E)$ is a hypergraph and $\boldsymbol{\phi}=(\phi_e)_{e\in E}$ such that each $\phi_e:\{0,1\}^e\to\mathbb{C}$ is a Boolean-variate complex-valued constraint function.
%
%
Furthermore, $\Phi=(V,E,\boldsymbol{\phi})$ is a  $(k,d)$-formula if
$\abs{e} \le k$ for every $e \in E$ and $\mathrm{deg}(v)=\abs{\{e\in E\mid v\in e\}} \le d$ for every $v \in V$.

The partition function for a Boolean CSP $\Phi=(V,E,\boldsymbol{\phi})$ of external field $\lambda$ is defined as:
\begin{align*}
    Z_\Phi(\lambda)=\sum_{\sigma \in \{0,1\}^V} \tp{\prod_{e \in E} \phi_e(\sigma|_e)} \lambda^{\norm{\sigma}_1}.
\end{align*}

In~\cite{Liu2019Ising}, Liu, Sinclair and Srivastava formulated the above partition function as counting hypergraph insects and gave a polynomial-time algorithm for such a partition function assuming its zero-freeness.
We will see that such a partition function is also subsumed by our framework.




Given a Boolean CSP $\Phi=(V,E,\boldsymbol{\phi})$, we can construct a dependency graph $G_\Phi=(V_\Phi,E_\Phi,\+L_\Phi)$ as follows.
\begin{enumerate}
    \item $V_\Phi=V$;
    \item for any distinct $u, v \in V_\Phi=V$, $\{u,v\} \in E_\Phi$ iff $\{u,v\} \subseteq e$ for some $e \in E$;
    \item for any $v\in V_\Phi=V$, its label $L_v = (\phi_e)_{e\in E, v \in e}$.
\end{enumerate}
Note that each constraint $\phi_e$ appears in labels of all $v\in e$, and the maximum degree of $G_\Phi$ is bounded by $\Delta\le (k-1)d$ for a $(k,d)$-formula $\Phi$.

Let $\+G_{k,d}$ denote the family of all such dependency graphs $G_{\Phi}$, where $\Phi=(V,E,\boldsymbol{\phi})$  is a $(k,d)$-formula,  and all their induced subgraphs.
Obviously, such $\+G_{k,d}$ is downward-closed.
%

Let $G\in\+G_{k,d}$. Without loss of generality, suppose that $G=G_{\Phi}[U]$ is the subgraph of the dependency graph $G_{\Phi}$ induced by $U\subseteq V$, where $\Phi=(V,E,\boldsymbol{\phi})$ is a Boolean CSP.

We define
\begin{align*}
    f_G(\lambda) = \sum_{\sigma \in \{0,1\}^U} \prod_{\substack{e \in E\\e\cap U\neq\emptyset}} \phi^U_e(\sigma|_{U \cap e}) \lambda^{\norm{\sigma}_1},
\end{align*}
where $\phi^U_e: U\cap e \rightarrow \mathbb{C}$ is defined as that for any $\tau \in \{0,1\}^{U \cap e}$,
\begin{align*}
    \phi^U_e(\tau) = \phi_e(\tau^*),
\end{align*}
where $\tau^*\in\{0,1\}^e$ extends $\tau$ and assigns all $v\in e\setminus U$ with 0.
It is easy to verify that such definition $f_G$ uses only the information stored in the dependency graph $G$, thus it is well defined.
Meanwhile, it is also easy to verify that $f_\cdot$ is multiplicative and $f_{G}(\lambda) = Z_\Phi(\lambda)$ if $G=G_\Phi$.

For $G=G_\Phi[U]$ where $\Phi=(V,E,\boldsymbol{\phi})$ and $U\subseteq V$, define $\lambda_{G,\ell}$, as
\[
    \lambda_{G,\ell} =
    \begin{cases}
    \prod_{\substack{e\in E\\e\cap U\neq\emptyset}} \phi_e^U(\*1_{U \cap e}), & \abs{U} = \ell\\
    0, & o.w.
    \end{cases}
\]
Each $\lambda_{G,\ell}$  can be determined in $\mathrm{poly}(k,d,\ell)$ time. 

Observe that,
\begin{align*}
    f_G(\lambda) &= 1+\sum_{\ell=1}^{\abs{U}} \sum_{\substack{\sigma \in \{0,1\}^U\\ \norm{\sigma}_1 = \ell}} \prod_{\substack{e\in E\\ e\cap U\neq\emptyset}} \phi^U_e(\sigma|_{U \cap e}) \lambda^\ell\\
    &= 1+ \sum_{\ell=1}^{\abs{U}} \sum_{\substack{S \subseteq U\\ \abs{S}=\ell}} \prod_{\substack{e\in E\\ e\cap S\neq\emptyset}} \phi^S_e(\*1_{S \cap e})  \lambda^\ell\\
    &= 1+ \sum_{\ell=1}^{+\infty} \tp{\sum_{\substack{S \subseteq U}} \lambda_{G[S],\ell} } \lambda^\ell.
\end{align*}

Therefore, $f_\cdot$ is a $(1,1)$-bounded on $\+G_{k,d}$. 

Applying \Cref{cr:algo}, we immediately obtain the following corollary.
Similar bound has been proved in~\cite{Liu2019Ising}, but here we only need to encode the problem using our framework.
\begin{corollary}\label{thm:classical-liu}
Let $\Omega \subseteq \mathbb{C}$ be a $\gamma$-good region for $\gamma\ge 1$.
%
%
For any $\delta \in (0,1)$, there is a deterministic algorithm such that given any $(k,d)$-formula $\Phi$ for Boolean CSP, provided that $Z_{\Phi}$ is $M$-zero-free on $\Omega$, for any external field $\lambda\in\Omega_{\delta}$ and error bound $\epsilon\in(0,1)$,
the algorithm outputs an estimation of $Z_\Phi(\lambda)$ within $\epsilon$-multiplicative error
    in $\widetilde{O}\tp{n\tp{\frac{M}{\delta\epsilon}}^C}$ time with $C=\frac{1}{\delta}\ln\tp{8\mathrm{e}kd(1+\gamma)}$.
\end{corollary}
Since such $Z_\Phi(\lambda)$ is a polynomial with finite degree, by \Cref{fact:convex} and \Cref{lm:poly-bound}, if $\Omega \subseteq \mathbb{C}$ is a convex region and $Z_\Phi$ does not vanish on a slightly larger region $\Omega' =\{(1+\delta)z \mid z \in \Omega\}$ for some constant gap $\delta \in \mathbb{R}^+$, then $M=O_{\delta}(n)$ and hence the algorithm in \Cref{thm:classical-liu} runs in ${\tp{\frac{n}{\epsilon}}^{O(\ln\tp{kd})}}$ time.



\subsection{A barrier to non-multiplicative functions}
Our framework based on functions $f_G$ induced by dependency graphs $G$  is fairly expressive. 
However, the current technique crucially relies on the multiplicative property of $f_{\cdot}$. 
It would meet a barrier when dealing with systems  lacking such multiplicative property.

We explain this  using a concrete example.
Consider the following generalization of ~\eqref{eq:classical-partition-function}:
\begin{align}
    \forall\beta\in\mathbb{C},\quad
        Z_{H,H'}(\beta)
        \triangleq \Tr{\exp(-\beta  H+H')}.\label{eq:general-quantum-partition-function}
\end{align}
Here, $H$ and $H'$ are two Hamiltonians in $\+V$. It encompasses the transverse Ising model and XXZ model.
Unfortunately, this partition function fails to fit in our framework due to the lack of multiplicative property even when $H'$ is a tensorized operator. 
For Hamiltonians $H_1,H_2$ such that $\supp(H_1) \cap \supp(H_2) = \emptyset$ and a tensorized operator $H' = \bigotimes_{v \in V} H'_v$, the following equality fails to hold in general:
\begin{align*}
    Z_{H_1+H_2,H'}(\beta) = \frac{1}{\Tr{\exp(-H')}} Z_{H_1,H'}(\beta) Z_{H_2,H'}(\beta).
\end{align*}
For example, for $\beta = 1$, $H_1 = I \bigotimes 
\tp{
\begin{matrix}
1&0\\
1&1\\
\end{matrix}
}
$, $H_2 = 
\tp{
\begin{matrix}
1&1\\
0&1\\
\end{matrix}
}
\bigotimes I
$
and $H' = \tp{
    \begin{matrix}
        1&0\\
        0&2
    \end{matrix}
    }
    \bigotimes 
    \tp{
    \begin{matrix}
        2&0\\
        0&1
    \end{matrix}
}$,
it holds that 
\begin{align*}
    \Tr{\exp(-H')} \approx 0.6869, \Tr{\exp(-H_{i}-H')} \approx 0.2416, i \in \{1,2\} \text{ and } \Tr{\exp(-H_1-H_2-H')} \approx 0.1316.
\end{align*} 
Hence, $\Tr{\exp(-H')} \Tr{\exp(H_1+H_2+H')} \approx 0.09040$ and $\Tr{\exp(-H_1-H')} \Tr{\exp(-H_2-H')} \approx 0.05837$. The main obstacle comes from the non-commutativity of Hamiltonians and it remains open to design a polynomial-time algorithm for such partition function assuming only zero-freeness.

\section{Efficient Coefficient Computing}\label{section-efficient-coefficient-computing}
In this section we prove  \Cref{thm:main}. First we need to establish the following lemma.

\begin{lemma}\label{lm:log-property}
%
Let $\+G$ be a downward-closed family of dependency graphs, and $f_{\cdot}$ be $(\alpha,\beta)$-bounded on $\+G$ for $\alpha,\beta\ge 1$.
    Recursively define the sequence $(\zeta_{H,i})_{H \in \+G, \ell\in \mathbb{N}^+}$ of complex numbers as follows:
    for any $H=(V_H,E_H,\+L_H)\in\+G$ and any $\ell\in\mathbb{N}^+$,
     \begin{align}\label{eq:lambda-prime}
        \zeta_{H,\ell} =
             \lambda_{H,\ell} - \sum_{s=1}^{\ell-1} \frac{s}{\ell} \sum_{\substack{S,T\subseteq V_H\\ S\cup T =V_H}} \zeta_{H[S],s} \lambda_{H[T],\ell-s}. 
    \end{align}
    It holds that $\zeta_{H,\ell} \neq 0$ only if $H$ is connected and $|V_H|\le\alpha \ell$.
    Moreover, for any $G=(V,E,\+L)\in\+G$, 
    \begin{align}\label{eq:log-induce-count}
        \log f_G(z)=\log f_G(0) + \sum_{\ell=1}^{+\infty} \tp{\sum_{S \subseteq V} \zeta_{G[S],\ell} }z^\ell.
    \end{align}
\end{lemma}

As in \cite{patel-polynomial, Liu2019Ising}, the following result of Borgs et al.~\cite{borgs2013left} is used.
\begin{fact}[Lemma 2.1 (c) in \cite{borgs2013left}]\label{lm:connected-graph}
    Let $G=(V,E)$ be a graph with maximum degree $\Delta$, $v \in V$ be a vertice and $\ell \in \mathbb{N}_{\ge 1}$.
    Then the number of connected subgraphs of size $\ell$ containing $v$ is at most $\frac{(\mathrm{e}\Delta)^{\ell-1}}{2}$.
\end{fact}


With this fact, we can enumerate all connected induced subgraphs $G[S]$ of $\abs{S} \le \ell$ vertices efficiently.

\begin{lemma}\label{lm:list-connect}
    There exists a deterministic algorithm which takes  a dependency graph $G=(V,E,\+L)$ on $n=|V|$ vertices with maximum degree $\Delta$ and a positive integer $\ell \in \mathbb{N}^+$ as input, and outputs
    \begin{align}\label{eq:def-L}
        \+C_{\le\ell} = \{S \subseteq V \mid \abs{S} \le \ell, G[S] \text{ is connected}\},
    \end{align}
 in time $\widetilde{O}(n(\mathrm{e}\Delta)^{\ell})$, where $\widetilde{O}(\cdot)$ hides  $\mathrm{poly}(\Delta, \ell,\log(n))$.
\end{lemma}

\begin{proof}
Let $\+C_{=\ell}^v$ denote the collection of such $S\subseteq V$ containing $v\in V$ that $|S|=\ell$ and $G[S]$ is connected. Clearly $ \+C_{\le\ell} =\bigcup_{\substack{v\in V\\ j\le \ell}}\+C_{=j}^v$.
Now construct each $\+C_{=\ell}^v$ inductively.
When $\ell=1$, $\+C_{=1}^v=\{ \{v\}\}$.
For $\ell\geq 2$, we enumerate all $S \in \+C_{=\ell-1}^v$ and $u \in V \setminus S$ such that $G[S \cup \{u\}]$ is connected, and include $S \cup \{u\}$ into $\+C_{=\ell}^v$.
It is easy to see that this correctly constructs $\+C_{=\ell}^v$.
%
By \Cref{lm:connected-graph}, $\abs{\+C^v_{=\ell}} \le (\mathrm{e}\Delta)^{\ell-1}/2$.
Representing each set $S$ as a string of vertices in $S$ sorted in increasing order of vertices, the set $\+C^v_{=\ell}$ can be stored by a standard dynamic data structure such as prefix trees, so that it takes $O(\Delta \ell (\mathrm{e}\Delta)^{\ell-1})$ time to iterate over all  $(S,u) \in \+C_{=\ell-1}^v\times V$ such that $G[S \cup \{u\}]$ may be connected, and for each such $(S,u)$ pair it takes $\mathrm{poly}(\Delta, \ell,\log n)$ time to check weather $G[S \cup \{u\}]$ is connected or $S\cup\{u\}$ is already in $\+C_{=\ell}^v$, and insert $S$ into $\+C_{=\ell}^v$ if necessary.
Overall, it takes $\widetilde{O}(n(\mathrm{e}\Delta)^{\ell})$ time to construct $\+C_{\le\ell}$.
\end{proof}


Combining the above algorithm with~\eqref{eq:lambda-prime}, we can compute coefficients $\zeta_{H,\ell}$ for  $\log f_G$ efficiently.
\begin{lemma}\label{lm:count-lambdap}
Let $\+G$ be a downward-closed family of dependency graphs, and $f_{\cdot}$ be $(\alpha,\beta)$-bounded on $\+G$ for $\alpha,\beta\ge 1$.
There exists a deterministic algorithm which takes a dependency graph $G=(V,E,\+L)\in\+G$ on $n=|V|$ vertices with maximum degree $\Delta$ and a positive integer $\ell \in \mathbb{N}^+$ as input, and outputs $(\zeta_{G[S],\ell})_{S \in \+C_{\le \alpha \ell}}$ within $\widetilde{O}\tp{n(8\mathrm{e}\beta\Delta)^{\alpha\ell}}$ time, where $\+C_{\le \alpha \ell}$ is defined in Eq.~\eqref{eq:def-L}.
\end{lemma}

The lemma follows by first enumerating all $S\in\+C_{\le \alpha \ell}$, which takes $\widetilde{O}\tp{n  (\mathrm{e}\Delta)^{\alpha \ell}}$ time according to Lemma~\ref{lm:list-connect}, and second  for every $S\in \+C_{\le\alpha \ell}$, taking $H=G[S]$ and computing $\zeta_{H,\ell}$ using a dynamic programming given by Eq.~\eqref{eq:lambda-prime}, which takes  $\widetilde{O}(8^{\alpha\ell}\beta^{\ell})$ time.

Let $\log f_G(z) = \log f_G(0) + \sum_{\ell=1}^{+\infty} g_{G,\ell} z^\ell$. Due to \Cref{lm:log-property}, $\zeta_{G[S],\ell}=0$ if $G[S]$ is disconnected or $|S|>\alpha\ell$, thus due to Eq.~\eqref{eq:lambda-prime}, the $\ell$-th Taylor coefficient  of $\log f_G$ is given by
\[
g_{G,\ell} = \sum_{S \subseteq V} \zeta_{G[S],\ell}
=
\sum_{S \in\+C_{\le\alpha\ell}} \zeta_{G[S],\ell}.
\]
Therefore, \Cref{thm:main} is proved.
It only remains to prove \Cref{lm:log-property}.

\begin{proof}[Proof of \Cref{lm:log-property}]

Fix an arbitrary $G\in\+G$, and consider $f_G$.
%
%
Let $\log f_G = \log f_G(0) + \sum_{\ell=1}^{+\infty} g_{G,\ell} z^\ell$ denote the Taylor's expansion of $\log f_G$ at the origin, and
$f_{G}(z) = f_G(0) + \sum_{\ell=1}^{+\infty} f_{G,\ell} z^\ell$ denote the Taylor's expansion of $f_G$ at the origin.
We prove by induction on $\ell\ge 1$ that
    \begin{align}\label{eq:requirement}
        g_{G,\ell} = \sum_{S \subseteq V} \zeta_{G[S],\ell},
    \end{align}
which implies~\eqref{eq:log-induce-count}.

For the induction basis, when $\ell =1$,  by \Cref{lm:Newton} we have $g_{G,1}=f_{G,1}$.
by the definition of bounded graph function in \Cref{def:good-function}, $f_{G,1}=\sum_{S\subseteq V} \lambda_{G[S],1}$;
and it follows from ~\eqref{eq:lambda-prime} that $\lambda_{G[S],1}=\zeta_{G[S],1}$.
Altogether, we have
    \begin{align*}
        g_{G,1} = f_{G,1} = \sum_{S \subseteq V} \lambda_{G[S],1} =\sum_{S \subseteq V} \zeta_{G[S],1}.
    \end{align*}

Now suppose that the induction hypothesis~\eqref{eq:requirement} holds for all $\ell' < \ell$.
We have
    \begin{align*}
\sum_{S \subseteq V}\zeta_{G[S],\ell}
&=
\sum_{S \subseteq V}\tp{\lambda_{G[S],\ell} - \sum_{s=1}^{\ell-1} \frac{s}{\ell} \sum_{\substack{L,R \subseteq V\\L \cup R = S}}
\zeta_{G[L],s}\cdot \lambda_{G[R],\ell-s}}\\
&=
\sum_{S \subseteq V} \lambda_{G[S],\ell} - \sum_{s=1}^{\ell-1} \frac{s}{\ell} \tp{\sum_{L \subseteq V} \zeta_{G[L],s}} \tp{\sum_{R \subseteq V} \lambda_{G[R],\ell-s}}\\
        &=f_{G,\ell} - \sum_{s=1}^{\ell-1} \frac{s}{\ell} g_{G,s} f_{G,\ell-s}&& \text{(Induction Hypothesis)}\\
        &=g_{G,\ell}. && \text{(\Cref{lm:Newton})}
    \end{align*}
This finishes the inductive proof of~\eqref{eq:requirement}.

Next, we prove that $\zeta_{H,\ell}=0$ if $H=(V_H,E_H,\+L_H)\in\+G$ is disconnected or $\abs{V_H}>\alpha\ell$.
Recall that $f_\cdot$ is $(\alpha,\beta)$-bounded, we have $\lambda_{H,\ell}=0$ for $\abs{V_H}>\alpha\ell$.
Then the fact that $\zeta_{H,\ell} = 0$ for $\abs{V_H}>\alpha\ell$ can be verified by induction on $\ell\ge 1$.
Specifically, by Eq.~\eqref{eq:lambda-prime},
\[
\zeta_{H,\ell} =
\lambda_{H,\ell} - \sum_{s=1}^{\ell-1} \frac{s}{\ell} \sum_{\substack{S,T\subseteq V_H\\ S\cup T =V_H}} \zeta_{H[S],s} \lambda_{H[T],\ell-s}.
\]
For the basis, $\zeta_{H,1}=\lambda_{H,1}=0$ when $|V_H|>\alpha$.
In general, observe that assuming $\abs{V_H}>\alpha\ell$, for any $S\cup T =V_H$, it must hold that either $|S|>\alpha s$ or $|T|>\alpha(\ell-s)$.
Therefore, assuming $\abs{V_H}>\alpha\ell$, $\zeta_{H,\ell} = 0$ follows from the induction hypothesis.

Finally, it remains to verify that $\zeta_{H,\ell} = 0$ if $H$ is disconnected, which follows from the multiplicative property of $f_{\cdot}$.
%
%
%
By contradiction, assume that $\zeta_{H,\ell} \neq 0$ for some disconnected $H\in\+G$.
Let $S^*\subseteq V_H$ be a minimal subset of $V$ such that $H[S^*]$ is disconnected and $\zeta_{H[S^*],\ell} \neq 0$.
Since  $H[S^*]$ is disconnected, there exist nonempty $L,R\subseteq S^*$ such that $L \cup R=S^*$ and $L,R$ are disconnected in $H[S^*]$.
Due to the multiplicative property of $f_{\cdot}$,  we have $f_{G[S^*]} = f_{G[L]}\cdot f_{G[R]}$. Therefore,
%
    \begin{align}\label{eq:indcount-1}
        g_{G[S^*],\ell} =g_{G[L],\ell}+g_{G[R],\ell} = \sum_{S\subseteq L} \zeta_{G[S],\ell} + \sum_{S \subseteq R} \zeta_{G[S],\ell},
    \end{align}
    where the first equation can be formally verified for any disjoint dependency graphs $G_1,G_2\in\+G$ and any $z$ in the neighborhood of the origin, such that for an arbitrary path $P$ in $\Omega$ connecting $z$ and  the origin,
    \begin{align*}
        \log f_{G_1 \cup G_2}(z) &= \log f_{G_1 \cup G_2}(0) + \int_{P} \frac{f'_{G_1 \cup G_2}(z)}{f_{G_1 \cup G_2}(z)} \,dz\ \\
        &= \log f_{G_1}(0) + \log f_{G_2}(0) + \int_{P} \tp{\frac{f'_{G_1}(z)}{f_{G_1}(z)} + \frac{f'_{G_2}(z)}{f_{G_2}(z)}} \,dz\ \\
        &= \log f_{G_1}(z) + \log f_{G_2}(z).
    \end{align*}
On the other hand,
    \begin{align}\label{eq:indcount-2}
        g_{G[S^*],\ell}  = \sum_{S \subseteq S^*} \zeta_{G[S],\ell} = \zeta_{G[S^*],\ell} + \sum_{S \subset S^*} \zeta_{G[S],\ell} =
        \zeta_{G[S^*],\ell} + \sum_{S \subseteq L} \zeta_{G[S],\ell} + \sum_{S \subseteq R} \zeta_{G[S],\ell},
    \end{align}
    where the last equation is due to the minimality of $S^*$.
    Combining~\eqref{eq:indcount-1} and~\eqref{eq:indcount-2}, we have $\zeta_{G[S^*],\ell}=0$, a contradiction.
\end{proof}

\section{Applications}\label{sec:applications}
In this section, we prove that any zero-free partition function of local Hamiltonians with bounded maximum degree is polynomial-time approximable if the temperature is close enough to 0.
This is formally stated by the following theorem.
Recall the definition of the partition function $Z_{H,\Op}$ induced by Hamiltonian $H$ under measurement $\Op$ in  \eqref{eq:quantum-pf1}.

\begin{theorem}\label{thm:quantum-concrete}
Let $k,d \in \mathbb{N}^+$, $h>0$, $\delta\in(0,1)$ and $\beta_0 = \frac{1}{5kdh}$.
There is a deterministic algorithm
such that given any $(k,d)$-Hamiltonian $H=\sum_{j=1}^m H_j$ on $n$ sites satisfying $\norm{H_j} \le h$ and tensorized measurement $\+O$,
for any temperature  $\beta \in \mathbb{D}_{(1-\delta)\beta_0}$ and error bound $\epsilon \in (0,1)$,
the algorithm outputs an estimation of  $Z_{H,\+O}(\beta)$ within $\epsilon$-multiplicative error
in $\widetilde{O}\tp{\tp{\frac{n}{\delta \epsilon}}^C}$ time with $C=\frac{1}{\delta}\tp{\ln 8\mathrm{e}d+ 3k \ln q}+1$.
\end{theorem}

It was established in \cite{harrow2020classical} that the partition function exhibits zero-freeness property when the inverse temperature $\beta$ is close to the $0$.
Similarly, we have the following lemma.
\begin{lemma}\label{lm:quantum-zerofree}
    Let $h \in \mathbb{R}^+$, $H=\sum_{j=1}^m H_j$ be a $(k,d)$-Hamiltonian on $n$ sites, and $\+O$ be a tensorized measurement.
    If $\abs{H_j} \le h$ for all $1 \le j \le m$, then for any $\beta \in \mathbb{D}_{\beta_0}$ where $\beta_0 = \frac{1}{5edkh}$, it holds that
    \begin{align*}
        \abs{\log \frac{Z_{H,\+O}(\beta)}{\mathbf{Tr}[\+O]}} \le n.
    \end{align*}
\end{lemma}

\Cref{thm:quantum-concrete} follows directly from \Cref{lm:quantum-zerofree} and \Cref{thm:quantum}.


Besides estimation of partition function, another related important computational problem is to sample according to the Gibbs state.

The quantum Gibbs state specified by Hamiltonian $H\in \Hspace$ and inverse temperature $\beta \in \mathbb{R}^+$ is:
\begin{align*}
    \rho_H(\beta) = \frac{\exp(-\beta H)}{Z_{H}(\beta)}.
\end{align*}

The classical distribution $\mu_{H,\beta}$ over $[q]^V$ is the quantum Gibbs state $\rho_H(\beta)$ after measurement in the computational basis, i.e.
\begin{align*}
   \forall \sigma\in[q]^V,\quad  \mu_{H,\beta}(\sigma) = \frac{Z_{H,\+O_{\sigma}}(\beta)}{Z_H(\beta)},
\end{align*}
where $\+O_\sigma=\ket{\sigma} \bra{\sigma}$.
Note that $\mu_{H,\beta}$ is a well-defined distribution over $[q]^V$.
To see this, first note that $\sum_{\sigma}\+O_{\sigma}=I$ is the identity matrix in $\Hspace$, and hence $\sum_{\sigma}Z_{H,\+O_{\sigma}}(\beta)=Z_H(\beta)$; and second, both $\+O_\sigma$ and $\exp(\beta H)$ are positive semidefinite since $H$ is Hermitian and $\beta\in\mathbb{R}^+$, and hence $Z_{H,\+O_{\sigma}}(\beta)=\Tr{\exp(\beta H)\+O_{\sigma}}\ge 0$.

In the same regime as \Cref{thm:quantum-concrete}, we have a polynomial-time approximate sampler from $\mu_{H,\beta}$, the classical distribution obtained after measurement of the quantum Gibbs state in the computational basis.

\begin{theorem}\label{thm:quantum-sampler}
Let $k,d \in \mathbb{N}^+$, $h>0$, $\delta\in(0,1)$ and $\beta_0 = \frac{1}{5kdh}$.
There is a randomized algorithm
such that given any $(k,d)$-Hamiltonian $H=\sum_{j=1}^m H_j$ on $n$ sites satisfying $\norm{H_j} \le h$,
for any temperature  $\beta \in \mathbb{D}_{(1-\delta)\beta_0}$ and error bound $\epsilon \in (0,1)$,
the algorithm outputs an approximate sample $\sigma \in [q]^V$  within $\epsilon$ total variation distance from the  distribution $\mu_{H,\beta}$,
in $\widetilde{O}\tp{\tp{\frac{n}{\delta\epsilon}}^{C}}$ time with $C=\frac{1}{\delta}\tp{2\log 8\mathrm{e}d+6k\log q} + 3$.
\end{theorem}


\begin{proof}
We leverage the algorithm in \Cref{thm:quantum-concrete} as a subroutine, and give the following classical algorithm for approximate sampling from $\mu_{H,\beta}$.

    Without loss of generality, we may assume that $V=[n]$.
    Let $\+M_j = \ket{j} \bra{j}$ for $1 \le j \le q$, and $\+M_{v,j} = \tp{\bigotimes_{\ell=1}^{v-1} I} \otimes \+M_j \otimes \tp{\bigotimes_{\ell=v+1}^n I}$.
    Our procedure for sampling $\sigma \in [q]^V$ is as follows.
    \begin{enumerate}
        \item Initialize $\+O$ with the identity operator on Hilbert space $\+H$;
        \item Iterate $v$ from $1$ to $n$;
        \item For each $j$ from $1$ to $n$, estimate $z_{v,j}=Z_{H,\+O_{v-1} \+M_{v,j}}(\beta)$ within $\epsilon_0 = \frac{\epsilon}{10n}$-multiplicative error.
        \item samples $j \in [q]$ proportional to $\tilde{z}_{v,j}$, the estimation of $z_{v,j}$, updates $\+O$ with $\+O \+M_{v,j}$, and assigns $\sigma(v)$ with $j$.
    \end{enumerate}
    Note that $\+O$ is a tensorized measurement during the process. Hence, \Cref{thm:quantum-concrete} guarantees an estimation of $z_{v,j}$ within $\epsilon_0$-multiplicative error in $\widetilde{O}(\tp{\frac{n}{\epsilon \delta}}^{C})$ time with $C=\frac{1}{\delta}\tp{2\log 8\mathrm{e}d + 6k \log q}+2$.
    Furthermore, note that for each configuration $\sigma \in [q]^V$,
    \begin{align*}
        \frac{\mathbf{Pr}[\sigma \text{ is generated}]}{\mu_{H,\+O}(\sigma)} = \prod_{v=1}^n \frac{z_{v,\sigma(v)}}{\sum_{j \in [q]} z_{v,j}} \frac{\sum_{j=1}^q \tilde{z}_{v,j}}{\tilde{z}_{v,\sigma(v)}},
    \end{align*}
    and for each $v \in V$ and $j \in [q]$,
    \begin{align*}
        1-\epsilon_0 \le \frac{\tilde{z}_{v,j}}{z_{v,j}} \le 1+\epsilon_0.
    \end{align*}
    Hence,
    \begin{align*}
        1-\epsilon \overset{(*)}{<}\tp{\frac{1-\epsilon}{1+\epsilon}}^n \le \frac{\mathbf{Pr}[\sigma \text{ is generated}]}{\mu_{H,\+O}(\sigma)} \le \tp{\frac{1+\epsilon_0}{1-\epsilon_0}}^n \overset{(\star)}{<} 1+\epsilon,
    \end{align*}
    where $(\star)$ follows from $\tp{\frac{1+\epsilon_0}{1-\epsilon_0}} \le (1+3\epsilon_0)^n < \exp(\frac{3}{10} \epsilon) < 1+\epsilon$, and $(*)$ follows from $(\star)$ and $\frac{1}{1+\epsilon} > 1-\epsilon$.
    Therefore, the total varaince distance between $\mu_{H,\+O}$ and the output from our sampler will differ at most $\epsilon$.

    We conclude the proof by observing that our algorithm calls the subrountine $O(nq)$ times with parameter $\epsilon_0 = \frac{\epsilon}{10 n}$, which takes $\widetilde{O}\tp{\tp{\frac{n}{\epsilon}}^C}$ time with $C=\frac{1}{\delta}\tp{2\log 8\mathrm{e}d+6k\log q}+3$ in total.
\end{proof}

\subsection{Proof of \Cref{lm:quantum-zerofree}}

We now give a proof of \Cref{lm:quantum-zerofree}.
This result extends the zero-freeness result proved in \cite{harrow2020classical} to the case that allows tensorized measurement.
We will see that the same inductive proof based on cluster expansion works for this more general case.


Let $H=\sum_{i=1}^m H_i$ be a $(k,d)$-Hamiltonian, $\+O = \bigotimes_{v \in V} \+O_v$ be a tensorized measurement, and $X \subseteq V$ be an arbitrary subset.

Define $H_X$ the Hamiltonian $H$ restricted on $X$ as
\begin{align*}
    H_X = \sum_{\substack{i \in [m]\\ \mathrm{supp}(H_i) \subseteq X}} H_i,
\end{align*}
and define the partition function $Z_{H,\+O}$ restricted on $X$ as
\begin{align*}
    Z^X_{H,\+O}(\beta) = \mathbf{Tr}_X[\exp(-\beta H_X) \+O_X],
\end{align*}
where $O_X = \bigotimes_{v \in X} \+O_v$. And $Z^\emptyset_{H,\+O}=1$. Here the subscript $X$ in $\mathbf{Tr}_X$ indicates that the operators act on the sites in $X$.

Moreover, recall the dependency graph $G_H=(U,E,\+L)$ defined in the proof of \Cref{thm:quantum}:
\begin{enumerate}
    \item $U=[m]$;
    \item $E=\{(x,y) \in U \times U \mid x \neq y, \mathrm{supp}(H_x) \cap \mathrm{supp}(H_y) \neq \emptyset\}$;
    \item $L_x = H_x$ for any $x \in U$.
\end{enumerate}

We are now ready to introduce the cluster expansions of partition functions.
The following lemma was an extension of \cite[Lemma 26]{harrow2020classical} with tensorized measurement $\+O$.

\begin{lemma}[high temperature expansion \cite{harrow2020classical}]\label{lm:expansion}
    Let $h \in \mathbb{R}^+$, $H=\sum_{i=1}^m H_i$ be a $(k,d)$-Hamiltonian, and $\+O$ be a tensorized measurement.
    If $\norm{H_i} \le h$, then the following holds for all $\Lambda \subseteq V$, $x \in \Lambda$ and $\beta \in \mathbb{D}_{\beta_0}$.
    \begin{align*}
        Z^\Lambda_{H,\observable}(\beta) = \mathbf{Tr}[\+O_x] Z^{\Lambda \setminus \{x\}}_{H,\+O}(\beta) + \sum_{\substack{S \subseteq [m]\\ \exists j\in S, x \in \mathrm{supp}(H_j)\\ G_H[S] \text{ is connected}}} W_S(\beta) Z^{\Lambda \setminus R_S}_{H,\observable}(\beta),
    \end{align*}
    where
\begin{align*}
W_S(\beta)
&= \sum_{l=\abs{S}}^{+\infty} \frac{(-\beta)^l}{l!}\sum_{\substack{(i_1,i_2,\ldots,i_l) \in S^l\\ \cup_{j=1}^l \{i_j\} = S}}\mathbf{Tr}_{R_S} \left[\prod_{j=1}^l H_{i_j} \observable_{R_S}\right],\\
R_S
&= \bigcup_{j \in S} \mathrm{supp}(H_j),\\
\beta_0
&= \frac{1}{\mathrm{e}(\mathrm{e}-1) dh}.
\end{align*}
\end{lemma}

\begin{proof}
    Without lose of generality, we assume that $\Lambda = V$.
    Directly from the definition, we have
    \begin{align*}
        Z_{H,\observable}(\beta) &= \mathbf{Tr}_{V} \left[\sum_{\ell=0}^{+\infty} \frac{(-\beta)^\ell}{\ell!} \tp{\sum_{j=1}^m H_j}^\ell \observable\right]\\
        &\overset{(\star)}{=}\sum_{\ell=0}^{+\infty} \frac{(-\beta)^\ell}{\ell!} \mathbf{Tr}_{V} \left[H_{V \setminus \{x\}}^\ell\observable + \sum_{\substack{(i_1,i_2,\ldots,i_\ell) \in [m]^\ell\\ \exists j \in [\ell], x \in \mathrm{supp}(H_{i_j})}} \prod_{j=1}^\ell H_{i_j} \observable\right]\\
        &=\mathbf{Tr}[\observable_x] Z^{V \setminus \{x\}}_{H,\observable}(\beta) + \sum_{\ell=0}^{+\infty} \sum_{\substack{(i_1,i_2,\ldots,i_\ell) \in [m]^\ell\\ \exists j \in [\ell], x \in \mathrm{supp}(H_{i_j})}} \frac{(-\beta)^\ell}{\ell!} \mathbf{Tr}_{V} \left[\prod_{j=1}^\ell H_{i_j} \observable\right],
    \end{align*}
    where $(\star)$ follows from the fact that $H=H_{V \setminus \{x\}} + \sum_{\substack{j \in [m]\\x \in \mathrm{su
    pp}(H_j)}} H_j$.

    For each $\ell$-tuple $\*i=(i_1,i_2,\ldots,i_\ell)$ satisfying that there exists $j^* \in [\ell]$ such that $x \in \mathrm{supp}(H_{i_{j^*}})$, let $S_{\*i} \subseteq \{i_1,i_2,\ldots,i_\ell\}$ be the minimum subset such that $i_{j^*} \in S_{\*i}$ and $\tp{\bigcup_{j \in S_{\*i}} \mathrm{supp}(H_j)} \cap \tp{\bigcup_{j \in \{i_1,i_2,\ldots,i_\ell\} \setminus S_{\*i}} \mathrm{supp}(H_j)} =\emptyset$.
    Note that
    \begin{align*}
        \tp{\prod_{j=1}^\ell H_{i_j}} \+O = \tp{\prod_{\substack{1 \le j \le \ell\\i_j \in S_{\*i}}} H_{i_j}} \+O_{R_S} \tp{\prod_{\substack{1 \le j \le \ell\\ i_j \not \in S_{\*i}}} H_{i_j}} \+O_{V \setminus R_S},
    \end{align*}
    where $R_S = \bigcup_{j \in S} \mathrm{supp}(H_j)$.

    Therefore,
    \begin{align*}
        & \sum_{\ell=0}^{+\infty} \sum_{\substack{(i_1,i_2,\ldots,i_\ell) \in [m]^\ell\\ \exists j \in [\ell], x \in \mathrm{supp}(H_{i_j})}} \frac{(-\beta)^\ell}{\ell!} \mathbf{Tr}_{V} \left[\prod_{j=1}^\ell H_{i_j} \observable\right]\\
        &=  \sum_{\substack{S \subseteq [m]\\ \exists j \in S, x \in \mathrm{supp}(H_{j})}} \sum_{\ell=0}^{+\infty} \sum_{\substack{\*i=(i_1,i_2,\ldots,i_k) \in [m]^k\\ S = S_{\*i}}} \frac{(-\beta)^\ell}{\ell!} \mathbf{Tr}_{V} \left[\prod_{j=1}^\ell H_{i_j} \observable\right]\\
        &= \sum_{\substack{S \subseteq [m]\\ \exists i\in S, x \in \mathrm{supp}(H_i)\\ G_H[S] \text{ is connected}}} \sum_{l=\abs{S}}^{+\infty} \sum_{r=0}^{+\infty} \frac{(-\beta)^{l+r}}{l! r!} \tp{\sum_{\substack{(i_1,i_2,\ldots,i_l) \in S^l\\ \cup_{j=1}^l \{i_j\} = S}}\mathbf{Tr}_{R_S} \left[\prod_{j=1}^l H_{i_j} \observable_{R_S}\right]}
        \tp{\sum_{\substack{(i_1,i_2,\ldots,i_r) \in [m]^l\\ \mathrm{supp}(H_{i_l}) \subseteq V \setminus R_S}} \mathbf{Tr}_{V\setminus R_S} \left[\prod_{j=1}^r H_{i_j} \observable_{V \setminus R_S}\right]}\\
        &= \sum_{\substack{S \subseteq [m]\\ \exists j\in S, x \in \mathrm{supp}(H_j)\\ G_H[S] \text{ is connected}}} \tp{\sum_{l=\abs{S}}^{+\infty} \frac{(-\beta)^l}{l!}\sum_{\substack{(i_1,i_2,\ldots,i_l) \in S^l\\ \cup_{j=1}^l \{i_j\} = S}}\mathbf{Tr}_{R_S} \left[\prod_{j=1}^l H_{i_j} \observable_{R_S}\right]} Z^{V \setminus R}_{H,\observable}(\beta)\\
        &= \sum_{\substack{S \subseteq [m]\\ \exists j\in S, x \in \mathrm{supp}(H_j)\\ G_H[S] \text{ is connected}}} W_S(\beta) Z^{V \setminus R_S}_{H,\observable}(\beta),
    \end{align*}
    where these identities hold by Fubini's theorem assuming that the summation converges absolutely. Hence,
    it remains to verify that $\sum_{\substack{S \subseteq [m]\\ \exists j \in S, x \in \mathrm{supp}(H_j)\\ G_H[S] \text{ is connected}}} \abs{W_S(\beta) Z^{V \setminus R}_{H,\observable}(\beta)} < +\infty$.

    Note that
    \begin{align}
        \abs{W_S(\beta)} &\le \sum_{l=\abs{S}}^{+\infty} \frac{\abs{\beta}^l}{l!}\sum_{\substack{(i_1,i_2,\ldots,i_l) \in S^l\\ \cup_{j=1}^l \{i_j\} = S}} \mathbf{Tr}_{R_S}\left[\abs{\prod_{j=1}^l H_{i_j} \observable_{R_S}}\right]\notag\\
        & \le \mathbf{Tr}_{R_S}[\observable_{R_S}] \sum_{l=\abs{S}}^{+\infty} \frac{(\abs{\beta} h)^l}{l!} \abs{\{(i_1,i_2,\ldots,i_l) \in S^l \mid \cup_{j=1}^l \{i_j\} = S\}}\notag\\
        &= \mathbf{Tr}_{R_S}[\observable_{R_S}] \sum_{l=\abs{S}}^{+\infty} \frac{(\abs{\beta} h)^l}{l!} \sum_{\substack{j_1,j_2,\ldots,j_{\abs{S}} \ge 1\\ j_1+j_2+\ldots+j_{\abs{S}} = l}} \frac{l!}{j_1! j_2! \ldots j_{\abs{S}}!}\notag\\
        &= \mathbf{Tr}_{R_S}[\observable_{R_S}] \tp{\exp(\abs{\beta} h) - 1}^{\abs{S}},\label{eq:estimate-w}
    \end{align}
    where the two inequalities follow from triangle inequality and H\"older's inequality respectively.
    Similarly,
    \begin{align*}
        \abs{Z_{V \setminus R_S}^{H,\observable}(\beta)} \le \mathbf{Tr}_{V \setminus R_S}\left[\observable_{V \setminus R_S}\right] \exp(\norm{\beta H_{V \setminus R_S}})
        \le \mathbf{Tr}_{V \setminus R_S} \left[\observable_{V \setminus R_S}\right] \exp(\abs{\beta}hd n),
    \end{align*}
    where the last inequality follows from the fact that $m \le d n$. Hence,
    \begin{align*}
        \sum_{\substack{S \subseteq [m]\\ \exists j\in S, x \in \mathrm{supp}(H_j)\\ G_H[S] \text{ is connected}}} \abs{W_S(\beta) Z^{V \setminus R_S}_{H,\observable}(\beta)}
        &\le \mathbf{Tr}[\observable] \exp(\abs{\beta}h d n) \sum_{\substack{S \subseteq [m]\\ \exists j\in S, x \in \mathrm{supp}(H_j)\\ G_H[S] \text{ is connected }}} \tp{\exp(\abs{\beta} h) - 1}^{\abs{S}}\\
        &\le \mathbf{Tr}[\observable] \exp(\abs{\beta}h d n) \sum_{\ell=0}^{+\infty} (\mathrm{e} d)^\ell  \tp{\exp(\abs{\beta} h) - 1}^\ell,
    \end{align*}
    where the last inequality follows from \Cref{lm:connected-graph} that there exists at most $(\mathrm{e} d)^\ell$ connected subgraph of size $\ell$ in dependency graph $G_H$ which contains vertice $j \in U$ such that $x \in \mathrm{supp}(H_j)$.
    Since $\abs{\beta} < \frac{1}{\mathrm{e}(\mathrm{e}-1)d h}$,
    \begin{align*}
        \mathrm{e} d \tp{\exp(\abs{\beta} h)-1} < \mathrm{e}(\mathrm{e}-1) d h \abs{\beta} < 1,
    \end{align*}
    concluding the proof of \Cref{lm:expansion}.
\end{proof}

At last, we need the following technical lemma \cite[Lemma 27]{harrow2020classical} in order to prove \Cref{lm:quantum-zerofree}.
\begin{lemma}[\cite{harrow2020classical}]\label{lm:aram-27}
    Let $H=\sum_{i=1}^m H_i$ be a $(k,d)$-Hamiltonian, and $\beta_0 = \frac{1}{5\mathrm{e}dkh}$, then for $\abs{\beta} < \beta_0$
    \begin{align*}
        \sum_{\substack{S \subseteq [m]\\ \exists j\in S, x \in \mathrm{supp}(H_j)\\ G_H[S] \text{ is connected}}} \tp{\mathrm{e}^{\abs{\beta} h}-1}^{\abs{S}} \exp(d h \mathrm{e}^2 \abs{\beta} \abs{R_S}) \le \mathrm{e}(\mathrm{e}-1)d h\abs{\beta}.
    \end{align*}
\end{lemma}

\begin{proof}[Proof of \Cref{lm:quantum-zerofree}]
    Let $\Lambda \subseteq V$.
    As observed in \cite{harrow2020classical}, it suffices to prove that removal of single site $x \in \Lambda$ can only produce bounded additive overhead to $\log Z_{H,\+O}^\Lambda(\beta)$. Formally, we are going to prove that, for any $x \in \Lambda$,
    \begin{align*}
        \log \abs{\frac{1}{\mathbf{Tr}[\observable_x]} \frac{Z^\Lambda_{H,\observable}(\beta)}{Z^{\Lambda \setminus \{x\}}_{H,\observable}(\beta)}} \le \mathrm{e}^2 d h \abs{\beta}, && \abs{\beta} < \beta_0.
    \end{align*}
    We will prove the result by induction on $\abs{\Lambda}$.
    By \Cref{lm:expansion},
    \begin{align*}
        \abs{\log \abs{\frac{1}{\mathbf{Tr}[\+O_x]} \frac{Z^\Lambda_{H,\+O}(\beta)}{Z^{\Lambda \setminus \{x\}}_{H,\+O}(\beta)} } } &= \abs{\log \abs{1+ \sum_{\substack{S \subseteq [m]\\ \exists j\in S, x \in \mathrm{supp}(H_j)\\ G_H[S] \text{ is connected}}} W_S(\beta) \tp{\frac{1}{\mathbf{Tr}[\+O_x]}\frac{Z^{\Lambda \setminus R_S}_{H,\observable}(\beta)}{Z^{\Lambda \setminus \{x\}}_{H,\+O}(\beta)}}}}\\
        &\le -\log \tp{1-\sum_{\substack{S \subseteq [m]\\ \exists j\in S, x \in \mathrm{supp}(H_j)\\ G_H[S] \text{ is connected}}} \abs{W_S(\beta)} \abs{\frac{1}{\mathbf{Tr}[\+O_x]}\frac{Z^{\Lambda \setminus R_S}_{H,\observable}(\beta)}{Z^{\Lambda \setminus \{x\}}_{H,\+O}(\beta)}}}\\
        &\overset{(*)}{\le} -\log \tp{1-\sum_{\substack{S \subseteq [m]\\ \exists j\in S, x \in \mathrm{supp}(H_j)\\ G_H[S] \text{ is connected}}} \abs{W_S(\beta)} \tp{\frac{1}{\mathbf{Tr}[\+O_{R_S}]} \exp(-\mathrm{e}^2 dh \abs{\beta})}}\\
        &\overset{(\star)}{\le} -\log \tp{1-\sum_{\substack{S \subseteq [m]\\ \exists j\in S, x \in \mathrm{supp}(H_j)\\ G_H[S] \text{ is connected}}} \tp{\exp(\abs{\beta} h)-1}^{\abs{S}} \exp(-\mathrm{e}^2 dh \abs{\beta})},
    \end{align*}
    where $(*)$ follows from induction hypothesis and $(\star)$ follows from ~\eqref{eq:estimate-w}.
    Together with \Cref{lm:aram-27}, we have
    \begin{align*}
        -\log \tp{1-\sum_{\substack{S \subseteq [m]\\ \exists j\in S, x \in \mathrm{supp}(H_j)\\ G_H[S] \text{ is connected}}} \tp{\exp(\abs{\beta} h)-1}^{\abs{S}} \exp(-\mathrm{e}^2 dh \abs{\beta})} \le - \log (1-\mathrm{e}(\mathrm{e}-1) dh\abs{\beta}) \le \mathrm{e}^2 dh \abs{\beta},
    \end{align*}
    where the last inequality follows from the fact that $- \log \tp{1-\frac{\mathrm{e}-1}{\mathrm{e}} y} \le y$ for all $y \in [0,1]$.
\end{proof}

\section{Acknowledgements}
This research was  supported by National Natural Science Foundation of China (Grant No. 61972191) and the Program for Innovative Talents and Entrepreneur in Jiangsu.

\bibliographystyle{alpha}
\bibliography{note}
\clearpage

\appendix

\section{Region Transformation}\label{sec:appendix-poly}
We give a proof of \Cref{lm:poly} which says that the convex region $\Omega \subseteq \mathbb{C}$ is a $1$-good region, so long as $\mathrm{dist}(\Omega,z)$ can be determined in $O(1)$ time for all $z \in \Omega$.

The following result proved in \cite[Lemma 2.2.3]{barvinok-comb} explicitly gives a polynomial that maps a disk with radius slightly larger than $1$ to a strip.
\begin{lemma}[\cite{barvinok-comb}]\label{lm:bar-transform}
    Let $\delta \in (0,1)$ be a constant.
    Define $q_\delta \in \mathbb{C}[z]$ as follows:
    \begin{align*}
        q_\delta(z) = \frac{1}{\sum_{k=1}^n \frac{C^k}{k}} \sum_{k=1}^n \frac{(C z)^k}{k},
    \end{align*}
    where $C = 1-\exp\tp{-\frac{1}{\delta}}$, $n = \left \lfloor \tp{1+\frac{1}{\delta}} \exp(1+\frac{1}{\delta})\right\rfloor$. Then for all $\abs{z} < \rho$, where $\rho = \frac{1-\exp\tp{-1-\frac{1}{\delta}}}{1-\exp\tp{-\frac{1}{\delta}}} > 1$,
    \begin{enumerate}
        \item $q_\delta(0)=0$, $q_\delta(1)=1$,
        \item $\mathbf{Re}\tp{q_\delta(z)} \in [-\delta,1+2\delta]$ and $\abs{\mathbf{Im}\tp{q_\delta(z)}} \le 2\delta$.
    \end{enumerate}
\end{lemma}

We now can prove \Cref{lm:poly}.
\begin{proof}[Proof of \Cref{lm:poly}]
    Without loss of generality, we assume that $\beta \neq 0$. Note that the polynomial $q_\delta$ defined in \Cref{lm:bar-transform} satisfies
    \begin{enumerate}
        \item $q_\delta(\mathbb{D}_{\rho}) \subseteq S_{1,4\delta}$;
        \item $q_\delta(0)=0$, $q_\delta(1)=1$.
    \end{enumerate}
    Therefore, we can set $p_{\beta,\delta}(z) = \beta q_{\delta'}(\rho' z)$, where $\delta_0=\frac{\delta}{4 \abs{\beta}}$ and $\rho'=\frac{1-\exp\tp{-1-\frac{1}{\delta'}}}{1-\exp\tp{-\frac{1}{\delta'}}}$.
    We conclude our proof by observing that:
    \begin{enumerate}
        \item $p_{\beta,\delta}(\mathbb{D}) \subseteq S_{\beta,\gamma}$;
        \item $p_{\beta,\delta}(0) = 0$, $p_{\beta,\delta}\tp{\frac{1}{\rho'}} = \beta$.
    \end{enumerate}
\end{proof}

Here we give a proof of \Cref{fact:convex}, which is a direct application of \Cref{lm:poly}.
\begin{proof}[Proof of \Cref{fact:convex}]
    The convexity of $\Omega$ implies that $\mathrm{dist}([z_l,z_r],\partial \Omega) = \mathrm{min}(\mathrm{dist}(z_l),\mathrm{dist}(z_r))$ for arbitrary complex values $z_l,z_r \in \Omega$, where $[z_l,z_r] = \{z_l+t(z_r-z_l) \mid t \in [0,1]\}$.
    Hence, for each $x \in \Omega$, we set $f_x = p_{x,\delta}$, a polynomial defined in \Cref{lm:bar-transform}.
    We conclude the proof by observing that the $k$-th coefficient of $p_{x,\delta}$ can be determined in $O(k)$ time.
\end{proof}

\end{document}